\documentclass[a4paper,UKenglish]{lipics-v2018}
\nolinenumbers
\pdfoutput=1
\usepackage{etex} % Fixes out-of-registers error. Probably some bad package.
\reserveinserts{28}

\usepackage{xcolor}
\definecolor{amazing}{RGB}{254,67,101}

\usepackage{bbm}
\usepackage{mathrsfs}
\usepackage{latexsym}
\usepackage{paralist}
\usepackage{microtype}

\usepackage{suffix}

\usepackage{xpunctuate}
\usepackage[all,british]{foreign} % Typesets abbreviations like e.g.
\redefnotforeign[ie]{i.e\xperiodafter} % i.e and e.g should not be italicized
\redefnotforeign[eg]{e.g\xperiodafter}

\usepackage{xfrac} % Nice in-text fractions
\usepackage[style=english,english=british]{csquotes}

\usepackage{amsmath, amsthm, amssymb, amsfonts} % AMS packages
\usepackage{thmtools, thm-restate}
\usepackage{mathtools}

\usepackage[full,small]{complexity}

\usepackage[final,notref,color]{showkeys} % Use option 'final' to remove

\usepackage{xargs}
\usepackage{afterpage} % Provides macro to execute commands 
\PassOptionsToPackage{dvipsnames,usenames,table}{xcolor}
\usepackage{tikz}
\usetikzlibrary{calc}
\usepackage{booktabs} % Better tables
\usepackage{longtable}
\usepackage{float}
\usepackage{wrapfig}
\usepackage{floatflt}
\usepackage{framed}
\usepackage{dcolumn}
\usepackage{adjustbox}

\usepackage[shortlabels]{enumitem} % 'shortlabels' enable shorthands like in the enumerate package 
\usepackage{xspace}
\usepackage{xifthen}
\usepackage{fixltx2e}
\usepackage{url}
\usepackage{scrtime}

\usepackage[longend,vlined]{algorithm2e}
\usepackage{xkvltxp}
\usepackage[draft,english,silent]{fixme}
\newcommand{\todo}[1]{\fxfatal{\color{red}#1}}

\usepackage{index}

\makeindex
\newindex{sym}{sym.idx}{sym.ind}{Symbol index}
\newindex{prob}{prob.idx}{prob.ind}{Problem index}

\newcommand\mathindex[1]{\index[sym]{\ensuremath{{#1}}}}
\WithSuffix\newcommand\mathindex*[1]{\index*[sym]{\ensuremath{{#1}}}}

\makeatletter
\newcommand{\raisemath}[1]{\mathpalette{\raisem@th{#1}}}
\newcommand{\raisem@th}[3]{\raisebox{#1}{$#2#3$}}
\makeatother

\newcommand\Problem[2][]{%
    \index[prob]{#2@\textsc{#2}}\textsc{#2}\xspace%
}

\newclass{\paraNP}{paraNP}

\newcommand\restr[2]{{% Restriction of functions, set families
  \left.\kern-\nulldelimiterspace % automatically resize the bar with \right
  #1 % the function
  \vphantom{\big|} % pretend it's a little taller at normal size
  \right|_{#2} % this is the delimiter
  }}

\def\sminor^#1{
    \preccurlyeq_{\mathrlap{\mathsf{m}}}^{#1}
} % Shallow minor relation
\def\ssminor^#1{%
    \mathbin{\dot\preccurlyeq_{{\mathrlap{\mathsf{m}}}}^{#1}}\,
} % Stable shallow minor relation
\def\stminor^#1{
    \preccurlyeq_{\mathrlap{\mathsf{t}}}^{#1}
} % Shallow minor relation
\def\sstminor^#1{%
    \mathbin{\dot\preccurlyeq_{{\mathrlap{\mathsf{t}}}}^{#1}}\,
} % Stable shallow top. minor relation

\def\grad_#1{\nabla\!_{#1}}
\def\sgrad_#1{{\dot\nabla}\!_{#1}}
\def\topgrad_#1{\widetilde \nabla\!_{#1}}
\def\stopgrad_#1{\dot{\widetilde\nabla}\!_{#1}}
\def\topomega_#1{\widetilde \omega_{#1}}

\def\colnum_#1{ \operatorname{col}_{#1} }
\def\wcolnum_#1{ \operatorname{wcol}_{#1} }
\def\adm_#1{ \operatorname{adm}_{#1} }

\newcommand{\bnd}{\ensuremath{{\partial}}}    % Boundary symbol for t-boundaried graphs

\newcommand{\tbnd}{%
  \mathchoice%
  {\raisebox{4.5pt}{{\scriptsize$\circ$}}\mkern-2mu} % displaystyle
  {\raisebox{4.5pt}{{\scriptsize$\circ$}}\mkern-2mu} % textstyle
  {\raisebox{4.1pt}{{\tiny$\circ$}}\mkern-3mu} % scriptstyle
  {\raisebox{4.5pt}{{\scriptsize$\circ$}}\mkern-2mu} % scriptscriptstyle
}

\renewcommand{\leq}{\leqslant}

\renewcommand{\geq}{\geqslant}

\renewcommand{\epsilon}{\varepsilon}

\renewcommand{\emptyset}{\varnothing}

\newlength{\convarrowwidth}
\usepackage{stmaryrd}

\newcommand{\widthm}[1]{ \mathbf{#1} } % Width measure style

\DeclareMathOperator{\width}{ \widthm{width} } % Generic width function for decompositions

\def\YYYY{{Y_0 \uplus Y_1 \uplus \cdots \uplus Y_\ell}} % I don't want to type this anymore
\WithSuffix\def\YYYY'{{Y'_0 \uplus Y'_1 \uplus \cdots \uplus Y'_{\ell'}}}

\def\ds{\mathop{\mathbf{ds}}}

\usepackage{pifont}% http://ctan.org/pkg/pifont
\newcommand{\yaay}{\kern4pt \ding{51} \kern-8pt \ding{51}}%
\DeclarePairedDelimiter\ceil{\lceil}{\rceil}

\newtheorem{prop}[theorem]{Proposition}
\newtheorem{observation}[theorem]{Observation}
\newtheorem*{claim}{Claim}

\renewcommand*\etal{\xperiodafter{\emph{et~al}}}

\newcommand*\varrule[1][0.4pt]{\leavevmode\leaders\hrule height#1\hfill\kern0pt}

\setlist[1]{labelindent=\parindent,leftmargin=*} 
\setlist{itemsep=0pt}
\setitemize[1]{label={\small\textbullet}}

\newenvironment{tightcenter}
 {\parskip=0pt\par\nopagebreak\centering}
 {\par\noindent\ignorespacesafterend}

\let\st\relax% ctable clashes with threeparttable
\usepackage{ctable}
\newlength{\RoundedBoxWidth}
\newsavebox{\GrayRoundedBox}
\newenvironment{GrayBox}[1]%
   {\setlength{\RoundedBoxWidth}{\textwidth-4.5ex}
    \def\boxheading{#1}
    \begin{lrbox}{\GrayRoundedBox}
       \begin{minipage}{\RoundedBoxWidth}%
   }{%
       \end{minipage}
    \end{lrbox}%
    \begin{tightcenter}%
    \begin{tikzpicture}%
       \node(Text)[draw=black!20,fill=white,rounded corners,%
             inner sep=2ex,text width=\RoundedBoxWidth]%
             {\usebox{\GrayRoundedBox}};
        \coordinate(x) at (current bounding box.north west);
        \node [draw=white,rectangle,inner sep=3pt,anchor=north west,fill=white] 
        at ($(x)+(6pt,.75em)$) {\boxheading};
    \end{tikzpicture}
    \end{tightcenter}\vspace{0pt}%
    \ignorespacesafterend
}    

\newenvironment{problem}[2][]{\noindent\ignorespaces%
                                \FrameSep=6pt%
                                \parindent=0pt%
                \vspace*{-.5em}
                \ifthenelse{\isempty{#1}}{%
                  \begin{GrayBox}{\textsc{#2}}%                
                }{%
                  \begin{GrayBox}{\textsc{#2} parametrised by~{#1}}%  
                }
                \index[prob]{#2@\textsc{#2}}%
                \newcommand\Prob{Problem:}%
                \newcommand\Input{Input:}%                        
                \newcommand\Param{Parameter:}%                        
                \begin{tabular*}{\textwidth}{@{\hspace{.1em}} >{\itshape} p{1.6cm} p{0.8\textwidth} @{}}%        
            }{
                \end{tabular*}%
                \end{GrayBox}%
                \vspace*{-.5em}
                \ignorespacesafterend
            }  
\newenvironment{problem*}[2][]{\noindent\ignorespaces%
                                \FrameSep=6pt%
                                \parindent=0pt%
                \vspace*{-.5em}
                \ifthenelse{\isempty{#1}}{%
                  \begin{GrayBox}{\textsc{#2}}%                
                }{%
                  \begin{GrayBox}{\textsc{#2} parametrised by~{#1}}%  
                }
                \index[prob]{#2@\textsc{#2}|textbf}%
                \newcommand\Prob{Problem:}%
                \newcommand\Input{Input:}%                        
                \begin{tabular*}{\textwidth}{@{\hspace{.1em}} >{\itshape} p{1.6cm} p{0.8\textwidth} @{}}%        
            }{
                \end{tabular*}%
                \end{GrayBox}%
                \vspace*{-.5em}
                \ignorespacesafterend
            }       

\renewenvironmentx{leftbar}[2][1=0.5pt, 2=5pt]{% 
  \def\FrameCommand{\vrule width #1 \hspace{#2}}\MakeFramed {\advance\hsize-\width \FrameRestore}}%
{\endMakeFramed}

\newlength{\wleft}  \newlength{\wright}

\definecolor{Maroon}{cmyk}{0, 0.87, 0.68, 0.32}
\definecolor{RoyalBlue}{cmyk}{1, 0.50, 0, 0}
\definecolor{Black}{cmyk}{0, 0, 0, 0}
\definecolor{White}{rgb}{1, 1, 1}

\makeatletter
\let\orgdescriptionlabel\descriptionlabel
\def\@savelabel{}
\renewcommand*{\descriptionlabel}[1]{%
  \let\orglabel\label
  \let\label\@gobble
  \phantomsection
  \def\@savelabel{#1}
  \edef\@currentlabel{{\def\hfil{}#1}}% Dirty trick because this adds a \hfil at the end
  \edef\@currentlabelname{#1}%
  \let\label\orglabel
  \orgdescriptionlabel{#1}%
}
\makeatother
\makeatletter
\def\namedlabel#1#2{\begingroup
   \def\@currentlabel{#1}%
   \label{#2}\endgroup
}
\makeatother

\usepackage{hyphenat}
\renewcommand{\th}{%
    \ifmmode% math mode
        ^\mathrm{th}%
    \else%
        \textsuperscript{th}\xspace%
    \fi%
}
\newcommand{\st}{%
    \ifmmode% math mode
        ^\mathrm{st}%
    \else%
        \textsuperscript{st}\xspace%
    \fi%
}
\newcommand{\nd}{%
    \ifmmode% math mode
        ^\mathrm{nd}%
    \else%
        \textsuperscript{nd}\xspace%
    \fi%
}
\newcommand{\rd}{%
    \ifmmode% math mode
        ^\mathrm{rd}%
    \else%
        \textsuperscript{rd}\xspace%
    \fi%
}

\def\Nesetril{Ne\v{s}et\v{r}il\xspace}

\def\Dvorak{Dvo\v{r}\'{a}k\xspace}

\def\Kral{Kr\'{a}l\xspace}

\usepackage{environ}% 
\NewEnviron{nofiitable}{\noindent\ignorespaces%
  
  \[\arraycolsep=1.4pt%
  \begin{array}{rcr p{1cm} rcr}
    \BODY
  \end{array}
  \]
}

\newcolumntype{m}{>{$}l<{$}}
\newcolumntype{M}{>{$\displaystyle}l<{$}} 

\newcolumntype{L}{l}  
\newcolumntype{C}{c}  
\newcolumntype{R}{r}  
\newcolumntype{X}{>{\global\let\currentrowstyle\relax}}
\newcolumntype{^}{>{\currentrowstyle}}

\title{Domination above $r$-independence: \newline does sparseness help?}

\titlerunning{Domination above $r$-independence}%optional, please use if title is longer than one line

\author{Carl Einarson}{Royal Holloway, University of London, UK}{einarson.carl@gmail.com}{}{}%mandatory, please use full name; only 1 author per \author macro; first two parameters are mandatory, other parameters can be empty.

\author{Felix Reidl}{Birkbeck, University of London, UK}{f.reidl@dcs.bbk.ac.uk}{https://orcid.org/0000-0002-2354-3003}{}

\authorrunning{C. Einarson and F. Reidl}%mandatory. First: Use abbreviated first/middle names. Second (only in severe cases): Use first author plus 'et al.'

\Copyright{Carl Einarson and Felix Reidl}%mandatory, please use full first names. LIPIcs license is "CC-BY";  http://creativecommons.org/licenses/by/3.0/

\subjclass{Theory of computation $\rightarrow$ Parameterized complexity and exact algorithms}% mandatory: Please choose ACM 2012 classifications from https://www.acm.org/publications/class-2012 or https://dl.acm.org/ccs/ccs_flat.cfm . E.g., cite as "General and reference $\rightarrow$ General literature" or \ccsdesc[100]{General and reference~General literature}.

%mandatory

\category{}%optional, e.g. invited paper

\relatedversion{}%optional, e.g. full version hosted on arXiv, HAL, or other respository/website

\supplement{}%optional, e.g. related research data, source code, ... hosted on a repository like zenodo, figshare, GitHub, ...

\funding{}%optional, to capture a funding statement, which applies to all authors. Please enter author specific funding statements as fifth argument of the \author macro.

\acknowledgements{}%optional

\EventEditors{John Q. Open and Joan R. Access}
\EventNoEds{2}
\EventLongTitle{42nd Conference on Very Important Topics (CVIT 2016)}
\EventShortTitle{CVIT 2016}
\EventAcronym{CVIT}
\EventYear{2016}
\EventDate{December 24--27, 2016}
\EventLocation{Little Whinging, United Kingdom}
\EventLogo{}
\SeriesVolume{42}
\ArticleNo{23}
\begin{document}

\maketitle
\begin{abstract}
  Inspired by the potential of improving tractability via
  gap- or above-guarantee parametrisations,
  we investigate the complexity of \Problem{Dominating Set}
  when given a suitable lower-bound witness. Concretely, we consider
  being provided with a maximal~$r$-independent set~$X$ (a set in which
  all vertices have pairwise distance at least~$r+1$) along the input
  graph~$G$ which, for~$r \geq 2$, lower-bounds the minimum size
  of any dominating set of~$G$. In the spirit of gap-parameters,
  we consider a parametrisation by the size of the `residual'
  set~$R := V(G)\setminus N[X]$.

  Our work aims to answer two questions: How does the constant~$r$ affect the
  tractability of the problem and does the restriction to sparse graph classes
  help here?
  For the base case~$r = 2$, we find that the problem is $\paraNP$-complete
  even in apex- and bounded-degree graphs. For~$r = 3$, the problem is~$\W[2]$-hard for
  general graphs but in $\FPT$ for nowhere dense classes and it admits a linear
  kernel for bounded expansion classes. For~$r \geq 4$, the parametrisation
  becomes essentially equivalent to the natural parameter, the size of the
  dominating set.
\end{abstract}

\newpage
\section*{Introduction}

The research of \emph{above/below guarantee} parameters as first used by Mahajan and
Raman~\cite{AboveGuarantee}  was an important step towards studying problems
whose natural parameters provided only trivial and unsatisfactory answers.
Case in point, the motivation for Mahajan and Raman was the observation that
every CNF-SAT formula with~$m$ clauses trivially has an assignment that
satisfies $\geq \ceil{m/2}$ clauses, thus question for the maximum number of
satisfied clauses is only interesting if~$k > \ceil{m/2}$, which of course
renders the parametrised approach unnecessary.  They therefore proposed to
study parametrisations `above guarantee': going with the previous example, we
would ask to satisfy~$\ceil{\sfrac{m}{2}} + k$ clauses or `$k$ above
guarantee'. After some isolated results in that direction
(\eg~\cite{VertexCoverAboveG,LinearArrangementAboveG})  the programme took up
steam after Mahajan~\etal presented several results and pointers in new
directions~\cite{AboveGuaranteeProgramme}
(\eg~\cite{MaxSatAboveTight,BetweennessAboveG,TernaryCSPAboveG,LinearArrangementAboveG}).
In particular, Cygan~\etal broke new ground for \textsc{Multiway Cut} and
\textsc{Vertex Cover} with algorithms that run in~$O^*(4^k)$ time, where~$k$
is the \emph{gap parameter} between an appropriate LP-relaxation and the
integral optimum~\cite{MultiwayCutAboveLP}. Lokshtanov \etal improved the
\textsc{Vertex Cover} case to~$O^*(2.3146^k)$ using a specialized branching 
algorithm~\cite{VCAboveLP}.

The latter result highlights an important realization: these alternative,
smaller parameters might not only provide the means to investigate problems
without `good' natural parameters, it might also provide us with faster
algorithms in practise! Gap- and above-guarantee parameters are attractive
because there is a reasonable chance that they are small in real-world
scenarios, something we often cannot expect from natural parameters.

To the best of our knowledge, so far no gap-parameter results are known for
domination problems and an above/below-guarantee result is only known in
bounded-degree graphs~\cite{AboveGuaranteeProgramme}. This is probably due to
the fact that there are no simple `natural' upper/lower bounds and in the case of
gap-parameters the LP-dualities do not provide much purchase. We therefore explore
this topic under the most basic assumptions: we are provided with a witness for a
lower bound on the domination number \emph{as input} and consider parametrisations
that arise from this additional information. In the case of
\Problem{Dominating Set}, the witness takes the form of a \emph{$2$-independent set},
that is, a set in which all vertices have pairwise distance at least three.
Note that this approach also captures a form of duality: the LP-dual of
dominating set describes a $2$-independent set, however, in general the two
optima are arbitrarily far apart. Recently, \Dvorak highlighted
this connection~\cite{DomsetDualitySparse} and proved that in certain sparse
classes the gap between the dual optima is bounded by a constant.

Thus, assume we are given a maximal $2$-independent set~$X$ alongside the input
graph~$G$. A parametrisation by~$|X|$ would go against the spirit of
gap-parameters, instead we parametrise by the size of  \emph{residual set}~$R :=
V(G)\setminus N[X]$, that is, all vertices that lie at distance two from~$X$
(since $X$ is maximal, no vertex can have distance three or more).
We choose this particular parameter for two reasons: 
\vspace*{-4pt}
\begin{enumerate}[(i)]
  \item For~$|R| = 0$ the problem is decidable in polynomial time since the
    domination number of  the graph is precisely~$|X|$. 
  \item The set~$X \cup R$ is
    a dominating set of~$G$. 
\end{enumerate} 
\vspace*{-4pt}
The first property is of course an
important pre-requisite for the problem to be in~$\FPT$ under this
parametrisation, while the second property guarantees us that the dominating
set size lies in-between~$|X|$ and~$|X| + |R|$.

Our first investigatory dimension is the constant~$r=2$ in the $2$-independent
set: intuitively, increasing the minimum distance between vertices in~$X$
increases the size of the parameter~$|R|$ and imposes more structure on the
input instance.  Our second dimension encompasses an approach that has been
highly successful in improving tractability of domination problems: restricting
the inputs to sparse graphs. While \Problem{Dominating Set} is $\W[2]$-complete
in general graphs, Alber \etal showed that it is fpt in planar
graphs~\cite{DomsetPlanarFPT}; Alon and Gutner later proved that assuming
degeneracy is sufficient~\cite{DomsetDegenerateFPT}. Philip, Raman, and Sikdar
extended this result yet further to graphs excluding a fixed bi-clique and also
proved that it admits a polynomial kernel~\cite{DomsetDegenerateKernel}. A
related line of research was the hunt for \emph{linear} kernels in sparse
classes. Beginning with such a kernel on planar graphs by Alber, Fellows, and
Niedermeier~\cite{DomsetPlanarKernel}, results on apex-minor free
graphs~\cite{BidimKernels}, graphs excluding a
minor~\cite{DomsetKernelHMinorFree} and classes excluding a topological
minor~\cite{DomsetKernelHTopFree} were soon proven. Recently, a linear kernel
for graphs of bounded expansion~\cite{DomsetBndExpKernel} (and an almost-linear
kernel for nowhere dense graphs~\cite{AlmostLinearKernelNowhereDense}) has subsumed all
previous results.\looseness-1

Our investigation of \Problem{Dominating Set} parametrised above an
$r$-independent set, for~$r \geq 2$, led us to the following results.
For~$r = 2$, the problem is $\paraNP$-complete already for~$|R| = 1$,
squashing all hope for an $\FPT$ or even~$\XP$ algorithm. This also holds
true if the inputs are restricted to sparse graph classes (apex-graphs/graphs
of maximum degree six).

For~$r = 3$, the problem is~$\W[2]$-hard in general graphs but admits
an~$\XP$-algorithm. In nowhere dense and bounded expansion classes, it is
fixed-parameter tractable. We further show, in the probably most technical
part of this paper, that it admits a linear kernel in bounded expansion
classes.

Finally, for~$r \geq 4$ the problem remains~$\W[2]$-hard in general graphs
and essentially degenerates to \Problem{Dominating Set} (hence, all the above
mentioned results in sparse classes translate in the parametrisation
above $r$-independence).

\section{Preliminaries}

A set~$X \subseteq V(G)$ is \emph{$r$-independent} if each pair of distinct
vertices in~$X$ have distance at least~$r{+}1$, thus an independent
set is $1$-independent.
We write~$N(v)$ and $N[v]$, respectively, for the neighbourhood and the closed
neighbourhood of a vertex~$v$. We extend this notation to sets as follows: for $X \subseteq V(G)$ we
let~$N(X)$ be all vertices not in $X$ that have a neighbour in~$X$ and~$N[X] :=
X \cup N(X)$. We let~$N^i(X)$ be all vertices not in $X$ that are at most
distance $i$ from any vertex in $X$ and we let~$N^i[X] = X \cup N^i(X)$.  A
vertex set~$Z \subseteq V(G)$ is \emph{dominated} by a set~$D \subseteq V(G)$ if
for every vertex~$z \in Z$ we have~$N[z] \cap D \neq \emptyset$, $D$ is then
called a \emph{$Z$-dominator}. We let~$\ds(G)$ denote the size of a minimum
dominating set of~$G$.

\begin{problem}{Dominating Set \textnormal{above $r$-independence}}
  \Input & A graph~$G$, a maximal $r$-independent set~$X \subseteq V(G)$, an
       integer~$p$. \\
  \Param & The size of the residual set~$R := V(G) \setminus N[X]$. \\
  \Prob & Does~$G$ have a dominating set of size~$p$?
\end{problem}

\noindent
Note that~$X \cup R$ is trivially a dominating set and that, for~$r
\geq 2$, it holds that~$\ds(G) \geq |X|$, thus we will
tacitly assume in the following that~$|X| \leq p \leq |X|+|R|$ since all
other instances are trivial.

We will frequently invoke the terms \emph{bounded expansion} and \emph{nowhere dense}
to describe graph classes. The definitions of these terms requires the introduction
of several concepts which will not be useful for the remainder of the paper, we
refer the reader to the book by \Nesetril and Ossona~de~Mendez~\cite{Sparsity}. In this
context, it is important to know that bounded expansion classes generalize most structurally
sparse classes (planar, bounded genus, bounded degree, $H$-minor free, $H$-topological minor free)
and nowhere dense classes contain bounded expansion classes in turn.
The following lemma and propositions for those two sparse graph classes will be needed
in the remainder of this paper:

\begin{lemma}[Twin class lemma~\cite{BndExpKernelsJournal,FelixThesis}]\label{lemma:twin}
  For every bipartite nowhere dense class there exists a 
  constant~$\omega$ and a function $f(s) = O(s^{o(1)})$ such that for 
  every member $G=(X,Y,E)$ of the class it holds that
  \begin{enumerate}
    \item $|\{ u \mid \deg(u) > 2\tau \}_{u \in Y}| \leq 2\tau\cdot |X|$, and
    \item $|\{ N(u) \}_{u \in Y}| \leq (\min\{ 4^\tau, \omega(e\tau)^\omega\} + 2\tau) \cdot |X|$.
  \end{enumerate}
  where $\tau = f(|X|) = O(|X|^{o(1)})$. If~$G$ is from a bounded expansion class, $\tau$ can be assumed a
  constant as well.
\end{lemma}

\noindent
We will also frequently invoke the following result regarding first-order (FO) 
model checking in bounded expansion and nowhere dense classes:

\begin{prop}[\Dvorak, \Kral, and Thomas~\cite{FOBndExp}]\label{prop:fo-check}
  For every bounded expansion class, the first-order model checking problem is solvable in linear fpt-time parametrised by the size of the input formula.
\end{prop}

\noindent
This result has since been extended to nowhere dense classes as well. Here,
\emph{almost linear fpt-time} means running time of the form~$O(f(k) \cdot
n^{1+o(1)})$ for some function~$f$.

\begin{prop}[Grohe, Kreutzer, and Siebertz~\cite{FONowhereDenseJournal}]\label{prop:fo-check-nowhere}
  For every nowhere dense class, he first-order model checking problem is solvable in almost linear fpt-time parametrised by the size of the input formula.
\end{prop}

\section{Above $2$-independence: hard as nails}

\noindent
In this section we will show that when we let $r=2$, we find that the problem
is $\paraNP$-complete for~$|R| = 1$, hence this parametrisation does not even
admit an $\XP$-algorithm. In the following we first present a reduction from
\Problem{3SAT} and then discuss how to modify it to reduce into sparse graph
classes.

Since $X$ is a (maximal) $2$-independent set, we know that each vertex in $R$
is a neighbour of some vertex in $N(X)$, otherwise we could add this vertex
to~$X$. Let us now describe the reduction. Let $\phi$ be $\Problem{3SAT}$-instance
with variables~$x_1,\ldots,x_n$ and clauses~$C_1,\ldots,C_m$. We construct~$G$
as follows (\cf Figure~\ref{fig:reduction-2-is}):
\begin{enumerate}
  \item For each variable~$x_i$, add a triangle with vertices~$x_i$, $t_i$, $f_i$.
  \item For each clause~$C_j$ add a vertex~$c_j$. If the variable~$x_i$ occurs
        positively in~$C_j$, add the edge~$c_j t_i$; if it occurs negatively, add
        the edge~$c_j f_i$.
  \item Add a single vertex~$y_1$ to the graph and connect it to each clause
        variable~$c_i$. Add two further vertices~$y_2, y_3$ and add the
        edges~$y_1y_2$ and~$y_2y_3$.
\end{enumerate}
We further set~$X := \{x_1,\ldots,x_n,y_1 \}$ as our~$2$-independent set;
notice that the only vertex not contained in~$N[X]$ is~$y_3$. Hence,
$R := \{y_3\}$.

\begin{figure}[tb]
  \centering
  \includegraphics[width=.9\textwidth]{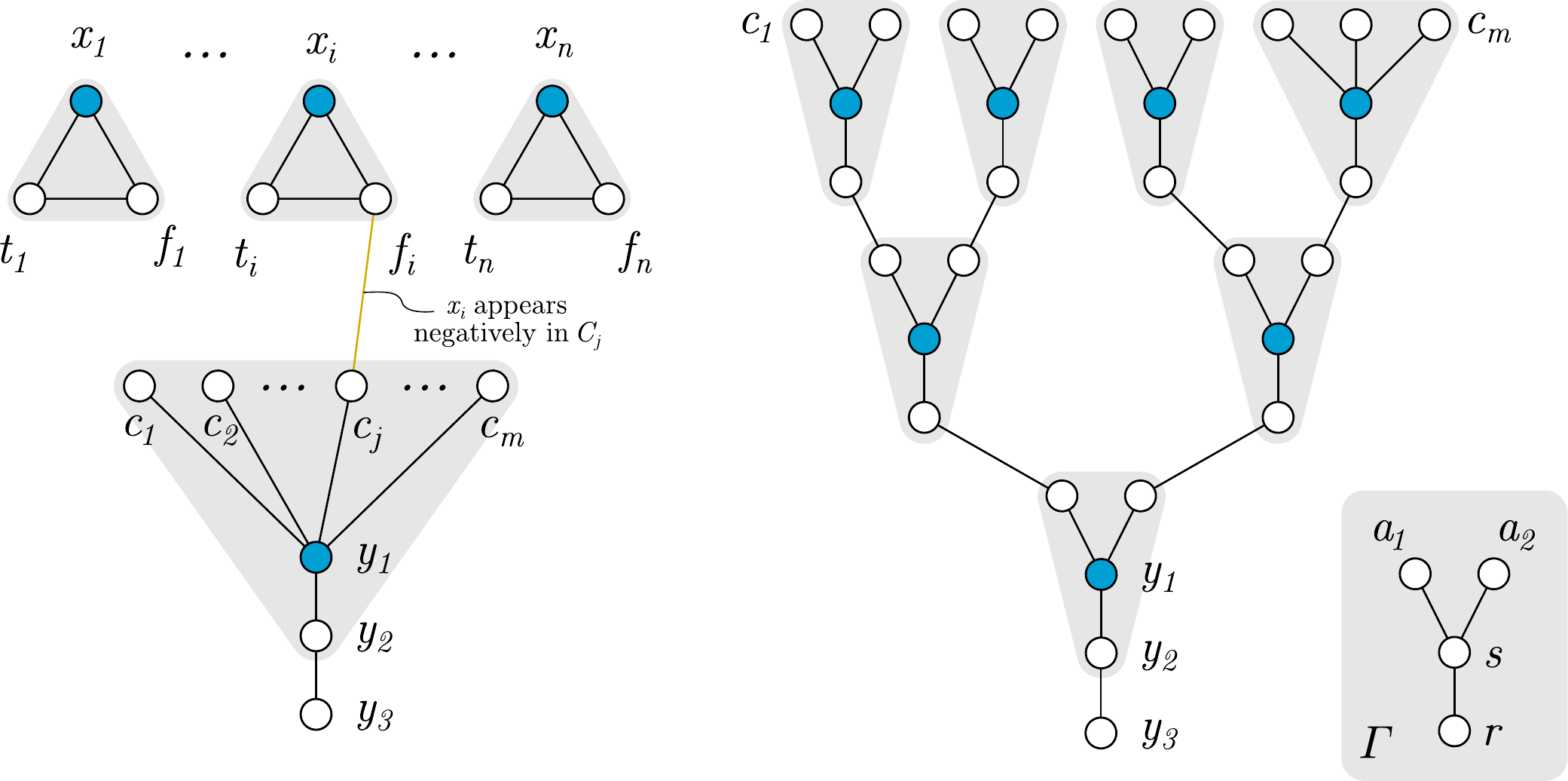}
  \caption{\label{fig:reduction-2-is}%
    Sketch of reduction from \Problem{3SAT} to \Problem{Dominating Set} above
    $2$-independent set. The left side shows the basic reduction, the right
    side shows the bounded-degree replacement gadget for the clause part, with
    the tree-gadget~$\Gamma$ highlighted on the bottom right. The $2$-independent
    set~$X$ is coloured blue and $N[X]$ is shaded in grey. In both
    constructions the set~$R$ consists only of~$y_3$. }
\end{figure}

\begin{lemma}\label{lemma:2-ind}
  $\phi$ is satisfiable iff $G$ has a dominating set of size~$|X|$.
\end{lemma}
\begin{proof}
  Assume~$\phi$ is satisfiable and fix one satisfying assignment~$I$. We construct
  a dominating set~$D$ as follows: if~$x_i$ is true under~$I$, add~$t_i$ to~$D$;
  otherwise add~$f_i$. Since~$I$ satisfies every clause of~$\phi$ the dominating
  set so far dominates every clause vertex and, of course, every variable gadget.
  The remaining undominated vertices are~$y_1, y_2, y_3$, thus adding~$y_2$ to~$D$
  yields a dominating set of~$G$ of size~$|X|$.

  In the other direction, assume that~$D$ is a dominating set of~$G$ of size~$|X|$.
  Since~$y_3$ is a pendant vertex we can assume that~$y_2 \in D$ (if~$y_3$ would
  be in~$D$ we could exchange it for~$y_2$). That leaves~$|X|-1 = n$ vertices in~$D$,
  precisely the number of variable-gadgets. Since every variable gadget must include
  at least one vertex of~$D$, we conclude the every such gadget contains precisely
  one dominating vertex. Since that depletes our budget, no other vertex is contained
  in~$D$.

  By the usual exchange argument we may assume that the dominating vertex in
  each variable gadget is either~$f_i$ or~$t_i$ and not~$x_i$ for~$1 \leq i
  \leq n$; hence the dominating vertices inside the variable gadgets encode a
  variable assignment~$I_D$ of~$\phi$. Finally, note that the clause vertices are
  not dominated by~$y_2$ and~$y_1$ is not contained in~$D$. Hence, they must be
  completely dominated by vertices contained in the variable gadgets. Then, by
  construction, the assignment~$I_D$ satisfies~$\phi$ and the claim follows.
\end{proof}

\noindent
We conclude that $\Problem{3SAT}$ many-one reduces to~\Problem{Dominating Set}
above $2$-independence already with~$|R| = 1$. We obtain the following two
corollaries that demonstrate that sparseness cannot help tractability here:

\begin{corollary}
  \Problem{Dominating Set} above $2$-independence is $\paraNP$-complete
  in apex-graphs.
\end{corollary}
\begin{proof}
  We use the above construction but reduce from a planar variant
  of~\Problem{3SAT}.  To ensure that we can construct variable-gadgets without
  edge crossings, we choose to reduce from Lichtenstein's \Problem{Planar
  3SAT} variant~\cite{SeparablePlanar3SAT} which ensures that the following
  graph~$G'$ derived from the \Problem{Planar 3SAT} instance~$\phi$ is planar:
  \begin{enumerate}
    \item Every variable~$x_i$ of~$\phi$ is represented by two literal vertices
          $t_i,f_i$ with the edge~$t_if_i \in G'$
    \item Each clause~$C_j$ is represented by a vertex~$c_j$. If the
          variable~$x_i$ occurs positively in~$C_j$, the edge $c_j t_i$
          exits; if it occurs negatively, the edge~$c_j f_i$ exists.
  \end{enumerate}
  To complete~$G'$ to~$G$ we have to add the vertices~$x_i$ and connect
  them to~$t_i, f_i$. This is clearly possible without breaking planarity
  (picture placing~$x_i$ on the middle of the line segment~$f_i t_i$ and
  moving it perpendicular by a small amount, then the edges~$x_i f_i$ and~$x_i t_i$
  can be embedded without crossing other edges). The vertices~$y_2, y_3$ can
  be placed anywhere; finally the vertex~$y_1$ will break planarity (the embedding
  does not guarantee that the clause vertices lie on the outer face of the graph)
  and we conclude that~$G$ is indeed an apex-graph.
\end{proof}

\begin{corollary}
  \Problem{Dominating Set} above $2$-independence is $\paraNP$-complete
  in graphs of maximum degree six.
\end{corollary}
\begin{proof}
  We reduce from~\Problem{(3,4)SAT} ($\NP$-hardness shown in~\cite{Bounded34SAT})
  in which every clause has size three and every variable occurs in at most
  four clauses.
  We use the above construction with one modification.
  Instead of connecting all clause vertices to one vertex we create a bounded-degree
  tree with the clause-vertices as its leaves.

  We begin by partitioning the clause vertices in pairs~$P_i$; if there is an odd number
  of vertices the last group will have three vertices.
  Then, for each group $P_i = \{c,c'\}$, we add two vertices $s^i, r^i$, connect
  the clause-vertices~$c$ and~$c'$ to $s^i$, and add the edge $s^ir^i$.
  We further add each $s^i$ to our $2$-independent set $X$ and then create a set $L_1$
  consisting of each $r^i$.
  Now we iteratively construct the next level of the tree, starting with $L := L_1$:
  \begin{enumerate}
    \item If~$L = \{r^1\}$, create a single vertex~$y_3$, connect it to~$r^1$ and finish,
          otherwise proceed with the next step.
    \item Partition the vertices in $L$ into~$\ell$ groups~$\{l_i,r_i\}$ of pairs.
          If $|L|$ is odd, the last group will be a triple~$\{l_i, c_i, r_i\}$ instead.
    \item For each group, create a tree-gadget $\Gamma^i$,
	        with vertices $\{a_1, a_2, s^i, r^i\}$ (and $a_3$ if the group contains a third vertex),
          and edges $l_ia_1, r_ia_2$, ($c_ia_3$),
          $a_1s^i$, $a_2s^i$, ($a_3s^i$), and $s^ir^i$.
    \item Add each $s^i$ to $X$, let $L$ now be the set of all $r^i$ (for $1 \leq i \leq \ell$)
          and continue with Step~1.
  \end{enumerate}
  Figure \ref{fig:reduction-2-is}, on the right side, shows an example of this
  construction. We note that, in the last tree-gadget, $s^i$ and $r^i$ are the
  same vertices as $y_1$ and $y_2$ respectively in the figure. We conclude the
  construction by adding each $x_i$ from the variable-gadgets to $X$ and
  setting $p := |X|$. Notice that the only vertex not contained in~$N[X]$
  is~$y_3$ and thus $R :=
  \{y_3\}$.

  Since each variable in~\Problem{(3,4)SAT} can be in up to four clauses, the
  maximum degree for $t_i$ and $f_i$ is six. All clause-vertices have degree
  at most four and all other types of vertices have a degree not higher than
  that, hence the claimed degree-bound holds. It is left to show that
  $\phi$ is satisfiable iff there is a dominating set of size $p = |X|$ in the graph.

  Let us assume that $\phi$ is satisfiable and fix one satisfying assignment $I$.
  We construct a dominating set as follows,
  beginning in the same way as in Lemma~\ref{lemma:2-ind}:
  if $x_i$ is true under $I$, add $t_i$ to $D$, otherwise add $f_i$.
  Since $I$ satisfies every clause of $\phi$ the dominating set so far dominates
  every clause vertex and every variable gadget.
  Now, the remaining undominated vertices are the tree-vertices,
  and our remaining budget is $|X|-n$ which is equal to the amount of tree-gadgets.
  Since every clause vertex is already dominated we can,
  for each tree-gadget $\Gamma^i = \{a_1, a_2, s^i, r^i \}$,
  add $r^i$ to the dominating set.
  This will dominate $s^i$ and the corresponding $a_1$ or $a_2$ in the tree-gadget below it,
  hence we can dominate all vertices of the graph within the budget~$|X|$.

  In the other direction, assume that $D$ is a dominating set of $G$ of size
  $|X|$. Since $y_3$ is a pendant vertex we can assume that $y_2 \in D$. Thus,
  in the last tree-gadget $r^i$ ($y_2$) is in the dominating set.
  This means that in order for $a_1$ and $a_2$ to be dominated,
  $r^i$ in both tree-gadgets above has to be
  in the dominating set. This holds for all tree-gadgets, all the way up to
  the clause vertices. Since we now have one vertex per tree-gadget in the
  dominating set this leaves $n$ vertices. Just as in Lemma \ref{lemma:2-ind},
  we note that, for each variable gadget, either $t_i$ or $f_i$ is in the
  dominating set. We know that the clause variables are not dominated by
  anything in the tree gadgets and thus must be dominated by the variable
  vertices. As stated in Lemma \ref{lemma:2-ind}, the dominating vertices
  inside the variable gadgets encode a variable assignment $I_D$ that
  satisfies $\phi$ and the claim follows.
\end{proof}

\section{Above $3$-independence: sparseness matters}

\subsection{$\W[2]$-hardness in general graphs}\label{sec:3-hardness}

\noindent
In the following we present a result for
\Problem{Dominating Set} above~$3$-independence,
namely that it is $\W[2]$-hard in general graphs.
We show this by reduction from \Problem{Colourful Dominating Set}
parametrised by the number of colours $k$:

\begin{problem}[$k$]{Colourful Dominating Set}
  \Input & A graph~$G$ with a vertex partition~$C_0,C_1,\ldots,C_k$. \\
  \Prob & Is there a set that dominates~$C_0$ and uses exactly one vertex
          from each set~$C_1,\ldots,C_k$?
\end{problem}

\noindent
It is easy to verify that \Problem{Colourful Dominating Set}
is $\W[2]$-hard by reducing from \Problem{Red-Blue Dominating Set}: 
we copy the blue set $k$ times and make each copy a
colour set $C_i$, $1 \leq i \leq k$, and let $C_0$ be the red set.
% applying colour-coding to instances of \Problem{Red-Blue Dominating Set}.

\begin{figure}[tb]
  \centering
  \includegraphics[width=.85\textwidth]{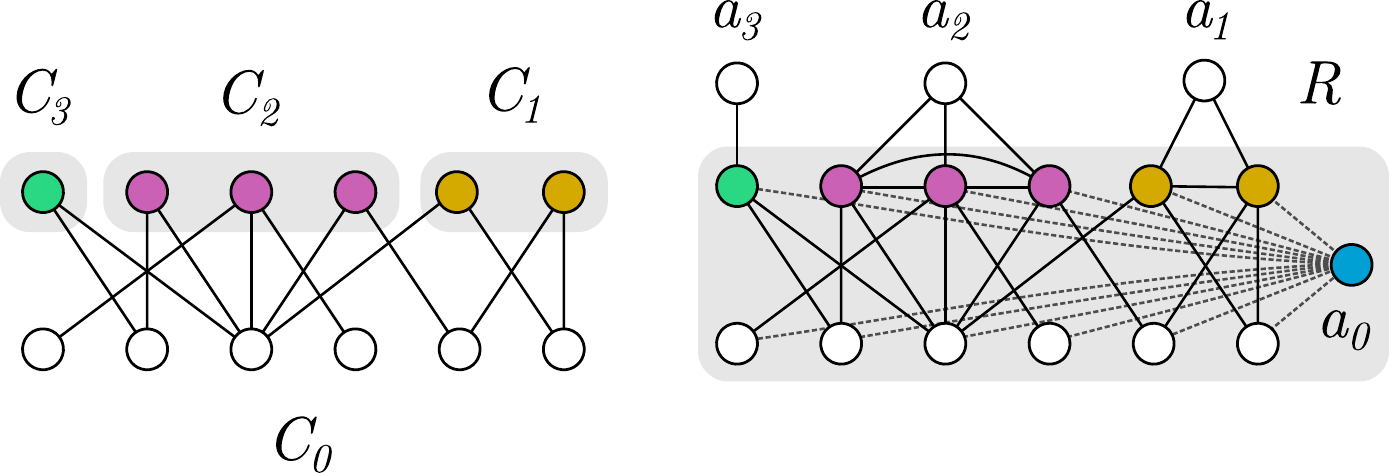}
  \caption{\label{fig:reduction-3-is}%
    Sketch of reduction from \Problem{Colourful Dominating Set} to \Problem{Dominating Set} above
    $r$-independent set for~$r \geq 2$. The set~$X$ contains only the vertex~$a_0$, the remaining
    vertices~$a_i$ are all contained in the residual~$R$.
  }
\end{figure}

\begin{lemma}\label{lemma:general-3-w2}
  \Problem{Dominating Set} above~$r$-independence is~$\W[2]$-hard for~$r \geq 2$.
\end{lemma}
\begin{proof}
  Let $I = (G,C_0,C_1,\ldots,C_k)$ be an instance of 
  \Problem{Colourful Dominating Set}. We construct an instance~$(G',X,k)$ of~\Problem{Dominating Set}
  above~$r$-independence as follows (\cf Figure~\ref{fig:reduction-3-is}):
  \begin{enumerate}
    \item Begin with~$G'$ equal to~$G$; then
    \item for each block $C_i$, $i \geq 1$, add edges to make $G[C_i]$ a complete graph
          and add an additional vertex~$a_i$ with neighbourhood~$C_i$; then
    \item add a vertex $a_0$ and connect it to all vertices in~$C_0 \cup C_1 \cup \ldots C_k$.
  \end{enumerate}
  Let $X=\{a_0\}$ be the $r$-independent set (which it clearly is for any~$r$)
  and thus $R = \{a_1, \ldots, a_k\}$. Note that the graph~$G'$ trivially has
  a dominating set of size $k+1$; the set $\{a_0, a_1, \ldots, a_k\}$. We now
  claim that $I$ has a colourful dominating set of size
  $k$ iff $G'$ has a dominating set of size $k$.

  Assume that $I$ has a solution of size~$k$ and fix one such colourful
  dominating set $D$. Note that~$D$ in~$G'$ a) dominates all of~$C_0$ (because
  it dominates~$C_0$ in~$G$) and b) dominates each set~$C_i \cup \{a_i\}$, $i \geq 1$
  (because~$D$ intersects each such $C_i$), and of course the vertex~$a_0$.
  Hence, $D$ is a dominating set of~$G'$ as well.

  In the other direction, assume that $D$ is a dominating set of~$G'$ of size $k$.
  Since each $a_i$, $i \geq 1$, is only connected to vertices
  in the corresponding set~$C_i$, at least one vertex from each set~$C_i \cup \{a_i\}$
  needs to be in~$D$. Since there are~$k$ such sets we conclude
  that~$D$ intersects each set~$C_i \cup \{a_i\}$ in precisely one vertex. We
  can further modify any such solution to not take the~$a_i$-vertices by taking
  an arbitrary vertex from~$C_i$ instead, thus assume that~$D$ has this form
  in the following. But then~$D$ is of course also a dominating set of size~$k$
  for~$G$, as  claimed.

\end{proof}

\subsection{Tractability in sparse graphs}

\noindent
In the following we present two positive results, namely that
\Problem{Dominating Set} above~$3$-independence is fpt in nowhere
dense classes and that it admits a polynomial kernel in bounded expansion
classes. The algorithm further implies an $\XP$ algorithm in general graphs.
The following annotated domination problem will occur as a subproblem:

\begin{problem}{Annotated Dominating Set}
  \Input & A graph~$G$, a subset~$Y \subseteq V(G)$, a
           collection of vertex sets~$R_1,\ldots,R_\ell$,
           and an integer~$k$. \\
  \Prob & Is there a set of size~$k$ that dominates~$V(G) \setminus Y$
          and contains at least one vertex in each~$R_i$?
\end{problem}

\noindent
In the following algorithm we will group vertices of~$N(X)$
according to their neighbourhood in~$R$ (or a subset of~$R$). We
will call those groups \emph{$R$-neighbourhood classes}.
We will write~$\gamma(G,Y,\{R_1,\ldots,R_\ell\})$ to denote the
size of an optimal solution of \Problem{Annotated Dominating Set}.

\begin{lemma}\label{lemma:above-3-fpt}
  \Problem{Dominating set} above~$3$-independence can be solved in linear fpt-time
  in any graph class of bounded expansion.
\end{lemma}
\begin{proof}
  First we guess the intersection~$D_R$ of an optimal solution (should it exist)
  with~$R$ in~$O(2^{|R|})$ time. Let~$R' \subseteq R$ be those vertices of~$R$
  that are not dominated by~$D_R$. Define
  \[
    \mathcal R := \{N(v) \cap R' \}_{v \in N(X)}
  \]
  as the neighbourhoods induced in~$R'$ by vertices in~$N(X)$. By
  Lemma~\ref{lemma:twin}, we have that~$|\mathcal R| = O(|R'|)$ since~$G$
  is from a class with bounded expansion (note that the partition into such
  neighbourhoods is computable in linear time using the partition-refinement
  data structure~\cite{PartitionRefinement}). Accordingly, in time~$O(2^{|\mathcal R|})
  = 2^{O(|R'|)}$, we can guess a subset~$\mathcal R' \subseteq \mathcal R$ such that
  an optimal solution covers exactly the neighbourhoods~$\mathcal R'$ (if~$\mathcal R'$
  does not cover~$R'$ we abort this branch of the computation). We are now left
  with the task of choosing vertices from~$N[X]$ to a) cover the
  neighbourhoods~$\mathcal R'$ and b) dominate the vertices in~$N[X]$ which are not
  dominated by~$D_R$.

  Let us introduce the following notation to ease our task: for a collection
  of $R$-neighbourhoods $\mathcal S \subseteq \mathcal R'$ and a set~$Y
  \subseteq N(X)$, let~$\mathcal S^{-1}(Y) := \{ Y_i \subseteq Y \mid N(y) =
  S~\text{for all}~y \in Y_i \}_{S \in \mathcal S}$. That is, $\mathcal
  S^{-1}(Y)$ contains those $R'$-neighbourhood classes in~$Y$ whose neighbourhood
  is contained in~$\mathcal S$.
  Let~$x_1,\ldots,x_\ell$ be an ordering of~$X$ and let~$H_i := G[N[x_i]]$; we
  will describe a dynamic-programming algorithm over the
  ordering~$x_1,\ldots,x_\ell$. Let~$T_i[\mathcal S]$ be the minimum size of a partial
  solution in $N[\{x_1,\ldots,x_i\}]$ that covers the neighbourhoods~$\mathcal S \subseteq \mathcal R'$
  and together with~$D_R$ dominates all of~$G[\bigcup_{1\leq j \leq i}
  N[x_j]]$. We initialize~$T_0[\mathcal S] := \infty$ for all~$\emptyset \neq
  \mathcal S \subseteq \mathcal R'$ and~$T_0[\emptyset] := 0$, then compute
  the following entries with the recurrence\footnote{A proof for the correctness 
  of the recurrence can be found in the Appendix}
  \[
    T_{i+1}[\mathcal S] :=
      \min_{\mathcal S_1 \cup \mathcal S_2 = \mathcal S}
      \Big( T_i[\mathcal S_1] + \gamma\big(H_{i+1}, N(D_R) \cap N[x_{i+1}], \mathcal S_2^{-1}(N(x_{i+1}))  \big) \Big).
  \]
  Note that~$\gamma(\ldots)$ is the minimum size of a set that
  dominates~$N[x_i] \setminus N(D_R)$ while choosing at least one vertex from
  each member of~$S^{-1}_2(N(x_i))$ (if $\emptyset \in S^{-1}_2(N(x_i)$ we assume that
  $\gamma(\ldots) = \infty$). The latter constraint corresponds to
  dominating the neighbourhoods of~$\mathcal S_2$ in~$R'$ by using vertices
  from~$N[x_i]$. Once the DP table~$T_\ell$ has been computed, the size of an
  optimal solution is the value in~$T_\ell[\mathcal R']$.

  It remains to be noted that every neighbourhood graph $H_i$ admits a
  dominating set of size one, hence the annotated dominating set has size at
  most~$|\mathcal S_2|+1 = O(|R|)$. Thus, the problem of finding an annotated
  dominating set is FO-expressible by a formula of size~$O(|R|)$ and we can
  solve the subproblem of computing~$\gamma(\ldots)$ in time~$O(f(|R|) \cdot
  |N[x_i]|)$ for some function~$f$ using Proposition~\ref{prop:fo-check}.
  As a result, we obtain a linear dependence on the
  input size (note that~$|E| = O(|V|)$) in the running time and thus the
  problem is solvable in linear fpt-time.\looseness-1
\end{proof}

\noindent
The same proof works for nowhere dense classes by applying
Proposition~\ref{prop:fo-check-nowhere} instead of Proposition~\ref{prop:fo-check}:

\begin{corollary}
  \Problem{Dominating set} above~$3$-independence can be solved in almost linear fpt-time
  in any nowhere dense class.
\end{corollary}

\noindent
We finally note that the algorithm described in the proof of
Lemma~\ref{lemma:above-3-fpt} only needs a black-box fpt-algorithm for
\Problem{Annotated Dominating Set} to run in fpt time, thus it is very likely
that \Problem{Dominating set} above~$3$-independence is in $\FPT$ for other
highly structured but not necessarily sparse graph classes. We can run the
same algorithm on general graphs to obtain an $\XP$-algorithm using the
bound~$|\mathcal R| \leq 2^{|R'|}$ and a simple brute-force in~$\XP$-time
on the \Problem{Annotated Dominating Set} subinstance during the DP.

\begin{corollary}
  \Problem{Dominating set} above~$3$-independence is in $\XP$.
\end{corollary}

\subsection{Kernelization in sparse graphs}

\noindent
Let us now set up the necessary machinery for the kernelization.
A \emph{boundaried graph}~$\tbnd G$ is a tuple~$(G,R)$ where~$G$ is a graph and~$R
\subseteq V(G)$ is the \emph{boundary}. We also write~$\partial \tbnd G$ to denote
the boundary. For a graph~$G$ and an induced
subgraph $H$, the boundary~$\partial_G H \subseteq V(H)$ are those vertices of~$H$ that
have neighbours in~$V(G)\setminus V(H)$. Thus for every subgraph~$H$ of~$G$
there is a naturally associated boundaried graph~$\bnd H = (H,\partial_G H)$.

For a boundaried graph~$\tbnd H$ and subsets~$A,B \subseteq \partial \tbnd H$
a set~$D \subseteq V(\tbnd H)$ is an~$(A,B)$-dominator of~$\tbnd H$ if~$D \cap
\partial \tbnd H = A$ and~$D$ dominates the set~$(V(\tbnd H) \setminus
\partial \tbnd H) \cup B$. We let~$\ds(\tbnd H, A, B)$ denote the size of a minimum
$(A,B)$-dominator of~$\tbnd H$.
A \emph{replacement} for~$H$ is a boundaried graph $(H',B)$ with~$H[B]
= H'[B]$. The operation of \emph{replacing~$H$ by~$H'$ in~$G$}, written
as~$G[H \to H']$, consist of removing the vertices~$V(H)\setminus B$ from~$G$,
then adding~$H'$ to~$G$ with~$B$ in~$H'$ identified with~$B$ in~$G$ (we assume
that the vertices~$V(H')\setminus B$ do not occur in~$G$).

\begin{lemma}\label{lemma:prot-replace}
  Let~$(G,X,p)$ be an instance of~\Problem{Dominating Set} above~$3$-independence
  where~$G$ is from a bounded expansion class.
  Let~$x \in X$, $H_x = G[N^2[x]]$ and~$R' = N^2[x] \cap R$.
  Then, in fpt-time with parameter~$|R'|$, we can compute a replacement~$H'_x$
  for~$H_x$ of size~$O(4^{|R'|}|R'|)$ such that
  $
    \ds(G[H_x \to H'_x]) = \ds(G)
  $.
  Moreover, the replacement~$H'_x$ is a subgraph of~$H_x$ and contains~$x$.
\end{lemma}
\begin{proof}
  For every pair of subsets~$A,B \in R'$ we compute
  a minimal~$(A,B)$-dominator~$S_{A,B}$ for~$\tbnd H_x = (H_x, R')$.
  Since this problem is expressible by an FO-formula of size~$O(|R'|)$, we can
  employ Proposition~\ref{prop:fo-check} to compute the set~$S_{A,B}$
  in linear fpt-time with parameter~$|R'|$. Note that $S_{A,B}$ will,
  besides the vertices in $A$, contain at most $|B| + 1$ additional
  vertices, since $|B|$ vertices suffice to dominate $B$ and the
  vertex $x$ dominates all of~$V(H_x) \setminus R$.

  Let~$S = \bigcup_{A,B \subseteq R'} S_{A,B} \cup R' \cup \{x\}$ be the union of all
  such computed solutions and the boundary. By construction, $|S| \leq
  (4^{|R'|}+1) |R'| + 1$. Since all~$(A,B)$-dominators live inside~$S$, we can
  safely remove the edges from~$H_x$ that do not have any endpoint in~$S$,
  call the resulting graph~$\tilde H_x$. Let~$Y := V(\tilde H_x) \setminus S$
  and let~$Y_1,\ldots,Y_t$ be a partition of~$Y$ into~$S$-twin classes
  (thus all vertices in~$Y_i$ have the same neighbourhood in~$S$ and no other
   class has this neighbourhood). Note that~$Y$ is an independent set in~$\tilde H_x$.

  Let~$H'_x$ be obtained from~$\tilde H_x$ by removing all but two
  representatives from every twin-class, denote those by~$Y'_i \subseteq Y_i$.
  Clearly, every~$(A,B)$-dominator of~$\tilde H_x$ is still
  an~$(A,B)$-dominator of~$H'_x$, we need to proof the other direction.
  Let~$D'$ be an~$(A,B)$-dominator of~$H'$. If~$D'\subseteq S$ then~$D'$ is
  still an $(A,B)$-dominator for~$H$. The same holds if~$D'$ intersects every
  twin-class in at most one vertex: each such class either has size one, in
  which case it has size one in~$\tilde H_x$ as well, or it has size two and
  the vertex not contained in~$D'$ is dominated by a vertex in~$S$. In either
  case, $D'$ dominates all of~$Y$ and thus is an~$(A,B)$-dominator of~$\tilde
  H_x$ and therefore of~$H_x$. Thus, assume that~$D'$ fully contains some class~$Y'_i$.
  Clearly, only one vertex of~$Y'_i$ is enough to dominate~$N(Y'_i)$ (and
  potentially vertices in~$B$), thus we can modify~$D'$ by picking the central
  vertex~$x$ instead to dominate the other vertex of~$Y'_i$ (which of course
  also dominates all of~$Y_i$ in~$\tilde H_x$). This can, of course, only happen
  once, otherwise we would reduce the size of the supposedly minimal set~$D'$.
  This leads us back to the previous case and we conclude that there exists
  an~$(A,B)$-dominator of~$\tilde H_x$ and thus~$H_x$ of equal size.

  We conclude that~$\ds(\tbnd H_x, A, B) = \ds(\tbnd H'_x, A, B)$ for every
  choice of~$A, B \subseteq R'$; which implies that 
  $
    \ds(G[H_x \to H'_x]) = \ds(G)
  $.
  Finally, by the twin-class lemma, $|H'_x| = O(|S|) = O(4^{|R'|}|R'|)$
  since~$\tau$ is a constant in bounded expansion classes.
\end{proof}

\begin{lemma}\label{lemma:prot-class-replace}
  Let~$(G,X,p)$ be an instance of~\Problem{Dominating Set} above~$3$-independence
  where~$G$ is from a bounded expansion class.
  Let~$R' \subseteq R$ be a subset and let~$X' \subseteq X$ be those vertices
  $x \in X$ with~$N^2(x) \cap R = R'$. Let~$H$ be the induced subgraph on
  $R' \cup \bigcup_{x \in X'} N[x]$.
  Assuming $X'$ is not empty we can, in fpt-time with parameter~$|R'|$,
  compute a replacement~$H'$ for~$H$ of size~$O(|R'|^{2|R'|+2} 4^{|R'|})$
  alongside an offset~$c$ such that
  $
    \ds(G[H' \to H]) = \ds(G) - c
  $.
  Moreover, the replacement~$H'$ is a subgraph of~$H$.
\end{lemma}
\begin{proof}
  Let~$X' := \{ x_1, \ldots, x_\ell \}$ and define the graphs~$H_i :=
  G[N^2[x_i]]$ for~$1 \leq i \leq \ell$. We first we apply
  Lemma~\ref{lemma:prot-replace} to replace every subgraph~$H_i$ by a subgraph~$H'_i$ of
  size~$O(4^{|R'|}|R'|)$ in linear fpt-time with parameter~$R'$. For
  simplicity, let us call the resulting graph~$G$ and relabel the
  graphs~$H'_i$ to~$H_i$ (note that Lemma~\ref{lemma:prot-replace} ensures
  that dominating set size does not change and that~$X$ is still a
  $3$-independent set of the resulting graph).

  For every~$x_i \in X'$ we compute a
  \emph{characteristic vector}~$\chi_i$ indexed by pairs of subsets of~$R'$
  with the following semantic: for~$A,B \subseteq R'$ we set~$\chi_i[A,B] = \ds(\tbnd H_i, A, B)$.
  If~$\chi_i[A,B]$ is larger than~$|A|+|B|+1$ we simply set~$\chi_i[A,B] = \infty$.

  Note that the constraint subproblem to compute~$\chi_i[A,B]$ is FO-expressible,
  by a formula of size~$O(|R'|)$, thus we can compute the vectors~$\chi_i$ for~$1 \leq i \leq \ell$
  in linear time fpt-time with parameter~$O(|R'|)$ using Proposition~\ref{prop:fo-check}.
  Let~$\equiv_\chi$ be the equivalence relation over the graphs~$H_i$ defined as
  $
    H_i \equiv_\chi H_j \iff \chi_i = \chi_j
  $
  and let~$\mathcal H := \{H_i\}_{1\leq i\leq\ell} / \equiv_\chi$ be the corresponding
  partition into equivalence-classes under~$\equiv_\chi$. Note
  that~$|\mathcal H| \leq (|R'|+1)^{2|R'|}$ since that is the number of possible
  characteristic vectors.

  The construction of~$H'$ is now simple: in every equivalence-class~$\mathcal
  C \in \mathcal H$ we select~$\min\{|\mathcal C|, |R|\}$ subgraphs and remove
  the rest; clearly~$H'$ is a subgraph of~$H$ of the claimed size. We let the
  offset~$c$ to be equal to the number of subgraphs removed in this way.

  We are left to show that~$\ds(G[H' \to H]) = \ds(G) - c$.
  Consider any minimal dominating set~$D$ for~$G$. We call a graph~$H_i$ \emph{interesting}
  under~$D$ if~$D$ intersects~$H_i$ in any vertex besides~$x_i$.
  \begin{claim}
    There exists a solution~$D'$ of size equal to~$D$ under which at
    most~$|R'|$ graphs per equivalence class~$\mathcal C \in \mathcal H$ are
    interesting.
  \end{claim}
  \begin{proof}
    Consider any such class~$\mathcal C \in \mathcal H$. If~$|\mathcal C| \leq
    |R'|$ we are done, so assume otherwise. Let~$R'' \subseteq R'$ be the set
    of vertices that are dominated through vertices in graphs contained
    in~$\mathcal C$ and select up to~$|R''|$ many graphs that together already
    dominate~$R''$; we let~$D'$ to be equivalent to~$D$ on these graphs. Any
    graph~$H \in \mathcal C$ \emph{not} selected in this way only needs to
    dominate itself and we add its centre vertex~$V(H) \cap X$ to~$D'$. Since
    every graph needs to intersect any dominating set in at least one vertex,
    $|D'| \leq |D|$ and since~$D$ is minimal we must have $|D'| = |D|$.
    Finally, only the~$|R''| \leq |R'|$ selected graphs are interesting
    under~$D'$, as claimed.
  \end{proof}

  \noindent
  Thus, let us assume in the following that~$D$ is such a minimal solution
  under which at most~$|R'|$ graphs per class~$\mathcal C \in \mathcal H$ are
  interesting. For such a solution~$D$ of~$H$, we construct a solution~$D'$
  of~$H'$ of size $|D| - c$ as follows. For a class~$\mathcal C \in \mathcal
  H$, let~$\mathcal C' \subseteq \mathcal C$ be those graphs that are contained
  in~$H'$ and let $\mathcal I \subseteq \mathcal C$ be the graphs that are
  interesting under~$D$. Since~$|\mathcal C| \leq |\mathcal C'|$, we can pair
  every graph~$H_i \in \mathcal C$ with a graph~$H_{i'} \in \mathcal C'$. Fix
  such a pair~$H_i, H_{i'}$, let~$A = R' \cap D$ and let~$B$ be those vertices
  of~$R'$ that are exclusively dominated vertices in~$H_i$. Since $\chi_i[A,B]
  = \chi_{i'}[A,B]$, there must exist a set of~$|D \cap V(H_i)|$ vertices
  in~$H_{i'}$ that dominate~$H_{i'}$ and~$B$. If we repeat this construction
  for every graph~$H_i \in \mathcal C$ with their respective pair in~$\mathcal
  C'$ and then pick the centre vertex $X \cap V(H_j)$ for all~$H_j \in
  \mathcal C' \setminus \mathcal C$, then the resulting set~$D'$ has
  size~$\leq |D| - c$ and dominates all of~$H'$.

  The same proof works in reverse if we start with a dominating set~$D'$
  of~$H'$ to construct a dominating set~$D$ of~$H$ with~$|D| \leq |D'| + c$;
  thus we conclude that~$\ds(G[H' \to H]) = \ds(G) - c$ and the claim follows.
\end{proof}

\noindent
\begin{theorem}\label{thm:exp-graph-linear-kernel}
  \Problem{Dominating Set} above~$3$-independence 
  has a linear kernel in bounded expansion graphs.
\end{theorem}

\begin{proof}
  Let~$(G,X,p)$ be the input instance and
  let~$x_1,\ldots,x_\ell$ be the members of the $3$-independent set~$X$.
  We define the graphs~$H_i := G[N^2[x_i]]$ and their respective
  $R$-neighbours~$R_i := N^2[x_i] \cap R$
  for~$1 \leq i \leq \ell$.
  Let~$\tau$ be as in Lemma~\ref{lemma:twin}. We partition the graphs~$\{H_i\}_{1 \leq i
  \leq \ell}$ into two sets~$\mathcal L$, $\mathcal S$ where $H_i \in \mathcal
  L$ iff~$|R_i| > 2\tau$; and~$\mathcal S$ contains all remaining graphs.

  Let us first reduce~$\mathcal S$. Let $\mathcal R := \{ R_i \mid H_i \in
  \mathcal S \}$ be the $R$-neighbourhoods of the graphs collected in~$\mathcal
  S$. By the twin class lemma, $|\mathcal R| \leq (4^\tau + 2\tau) |R|$,
  however, for each member~$R' \in \mathcal R$ we might have many
  graphs in~$\mathcal S$ that intersect~$R$ in precisely this set~$R'$.

  Fix~$R' \in \mathcal R$ for now and let~$\mathcal S[R']$ be those graphs
  of~$\mathcal S$ that intersect~$R$ in~$R'$. Let~$H_{R'} := G[\bigcup_{H \in
  \mathcal S[R']}V(H)]$ be the joint graph of the subgraphs in~$\mathcal
  S[R']$. Since~$|R'| \leq 2\tau$, and~$\tau$ is a constant depending only on
  the graph class, we can apply Lemma~\ref{lemma:prot-class-replace} to
  compute a replacement~$H'_{R'}$ of size~$O(|R'|^{2|R'|+2} 4^{|R'|})$
  with offset~$c$ in polynomial time. We apply the replacement~$G[H_{R'} \to H'_{R'}]$
  and decrease~$p$ by~$c$. Repeating this procedure for all~$R' \in \mathcal R$
  yields a graph~$G_1$ (a subgraph of~$G$) in which the small graphs~$\mathcal S$ in total contain
  at most
  $
    |\mathcal R| \cdot O(|\tau|^{2\tau+2} 4^{\tau})
    = O(|R|)
  $
  vertices, a $3$-independent set~$X' \subseteq X$ of~$G_1$, and a new input~$p'$.
  This concludes our reduction for~$\mathcal S$.

  Let us now deal with~$\mathcal L$ in~$G_1$. By the twin class lemma, $|\mathcal L| \leq
  2\tau |R|$. But then $|X'|$ has size bounded in $O(|R|)$ and we conclude that the
  size of a minimal dominating set for $G_1$ is bounded by $O(|R|)$, hence we assume
  that $p'$ is bounded by $O(|R|)$ (otherwise the instance is positive and we can
  output a trivial instance). We now apply the existing linear 
  kernel~\cite{DomsetBndExpKernel} for \textsc{Dominating Set} to $G_1$. The output
  of the kernelization is a subgraph of the original graph, hence we collect the
  remaining vertices of the $3$-independent set $X'' \subseteq X'$ in order to output
  a well-formed instance $(G'', X'', p'')$. This concludes the proof.
\end{proof}

\section{Above $4$-independence: simple domination}

\noindent
In the case where the lower-bound set~$X$ is $4$-independent we now have that
for distinct~$x,x' \in X$ it holds that~$N^2(x) \cap N^2(x') = \emptyset$,
thus for each~$x \in X$ the set~$N^2(x) \cap R$ can only be dominated from~$R$
and~$N(x)$. Let us call an instance~$(G,X)$ \emph{reduced} if for every~$x
\in X$ the intersection $N^2(x) \cap R$ is non-empty. We can easily pre-process
our input instance to enforce this property: if such an~$x$ would exist we can
simply remove~$N[x]$ from~$G$ to obtain an equivalent instance. In a reduced
instance the parameter~$|R|$ is necessarily big compared to~$|X|$:\looseness-1

\begin{observation}
  For every reduced instance~$(G,X)$ of \Problem{Dominating Set} above $4$-independence
  it holds that~$|R| \geq |X|$.
\end{observation}

\begin{corollary}
  Let~$\mathcal G$ be a graph class for which \Problem{Dominating Set} is
  in~$\FPT$. Then \Problem{Dominating Set} above $4$-independence is in~$\FPT$
  for~$\mathcal G$ as well.
\end{corollary}

\begin{corollary}
  Let~$\mathcal G$ be a graph class for which \Problem{Dominating Set} admits
  a polynomial kernel. Then \Problem{Dominating Set} above $4$-independence
  admits a polynomial kernel for~$\mathcal G$ as well.\looseness-1
\end{corollary}

\noindent
On the other hand, we showed in Section~\ref{sec:3-hardness} that the
problem remains~$\W[2]$-hard for any~$r \geq 2$ in general graphs.

\section{Conclusion}

\noindent
We considered \Problem{Dominating Set} parametrised by the residual of a
given~$r$-independent set and investigated how the value of~$r$ and the choice
of input graph classes affect its tractability. We observed that the
tractability does improve from~$r=2$ to~$r=3$ as it goes from being
$\paraNP$-complete to `merely` $\W[2]$-hard and at least admits an
$\XP$-algorithm. Larger values of~$r$, however, do not increase the
tractability as the problem becomes essentially equivalent to
\Problem{Dominating Set}.

If we consider sparse classes (bounded expansion and nowhere dense), the
improvement in tractability from~$r=2$ to~$r=3$ is much more pronounced;
changing from~$\paraNP$-complete to $\FPT$ and even admitting a linear
kernel in bounded expansion classes. We very much believe that the kernel
can be extended to nowhere dense classes, but leave that quite technical
task as an open question.

\smallskip
\noindent{\small \textbf{Acknowledgments}: We thank our anonymous reviewer for helpfully pointing
out how to achieve a linear kernel in Theorem~\ref{thm:exp-graph-linear-kernel}.}

\bibliographystyle{plainurl}% the recommnded bibstyle
\bibliography{biblio,conf}

\newpage
\section*{Appendix}

We claim that that~$T_i[\mathcal S]$ is the size of a partial
solution in $N[\{x_1,\ldots,x_i\}]$ that covers the 
neighbourhoods~$\mathcal S \subseteq \mathcal R'$
and together with~$D_R$ dominates all of~$G[\bigcup_{1\leq j \leq i}
N[x_j]]$.

As the inductive base, we set~$T_0[\mathcal S] := \infty$ for all~$\emptyset \neq
\mathcal S \subseteq \mathcal R'$ and~$T_0[\emptyset] := 0$. Clearly, a non-empty
$\mathcal S$ cannot possibly be dominated (note that since we guessed $D_R$, no vertex
in $R'$ is contained in the dominating set). The cost of dominating nothing in $R'$ is of course $0$, hence
the choice of $T_0[\emptyset]$. We conclude that the tables of the base cases 
describe the claimed quantity.

Since the base cases hold, let us assume that the table $T_i[\mathcal S]$ for
all $\mathcal S \subseteq \mathcal R'$ describe the claimed partial solutions.
Recall that the recurrence is given by
\[
  T_{i+1}[\mathcal S] :=
    \min_{\mathcal S_1 \cup \mathcal S_2 = \mathcal S}
    \Big( T_i[\mathcal S_1] + \gamma\big(H_{i+1}, N(D_R) \cap N[x_{i+1}], \mathcal S_2^{-1}(N(x_{i+1}))  \big) \Big).
\]
Let $D_{i+1} \subseteq N[\{x_1,\ldots,x_{i+1}\}]$ be a minimal set which 
covers $\mathcal S$ and dominates all of $G[\bigcup_{1\leq j \leq i+1} N[x_j]]$. 
We argue in the following that $T_{i+1}[\mathcal S] = |D_{i+1}|$.

For the first direction, 
let $D_i := D_{i+1} \cap N[\{x_1,\ldots,x_i\}]$ be the dominators of $D_{i+1}$ 
that lie outside of $H_{i+1}$. Let then $\mathcal S_1 \subseteq \mathcal S$ 
contain those sets of $\mathcal S$ that are covered by vertices in $D_i$
and let $\mathcal S_2 \subseteq \mathcal S$ be those that are covered by
$D_{i+1}\setminus D_i$. Since $D_{i+1}$ covers all of $\mathcal S$, we know
that $\mathcal S_1 \cup \mathcal S_2 = \mathcal S$ and hence this pair $\mathcal S_1,
\mathcal S_2$ appears in the above recurrence. By induction, we know that
$T_i[\mathcal S_1] = |D_i|$. Since $D_{i+1}$ covers $\mathcal S_2$ and
dominates $H_{i+1} \setminus N(D_R)$, it follows that 
\[
  \gamma\big(H_{i+1}, N(D_R) \cap N[x_{i+1}], \mathcal S_2^{-1}(N(x_{i+1}) \big)
  \leq |D_{i+1} \setminus D_i|.
\]
We conclude that 
therefore $T_{i+1}[\mathcal S] \leq |D_i| +  |D_{i+1} \setminus D_i| = |D_{i+1}|$.

In the other direction, consider any pair $\mathcal S_1, \mathcal S_2$ for which the right
hand side of the recurrence is minimized. By induction, we know that there exists
a set $D_i$ of size $T_i[\mathcal S_1]$ which covers $\mathcal S_1$ and dominates,
together with $D_R$, all of $G[\bigcup_{1\leq j \leq i} N[x_j]]$. Let $D'_{i+1}$
be a minimal solution to the \textsc{Annotated Dominating Set} instance
$\big(H_{i+1}, N(D_R) \cap N[x_{i+1}], \mathcal S_2^{-1}(N(x_{i+1}))  \big)$. By
definition, $D'_{i+1}$ covers $\mathcal S_2$ and dominates all of
$H_{i+1} \setminus N(D_R)$. But then $D_i \cup D'_{i+1}$ is a subset of
$[\{x_1,\ldots,x_{i+1}\}]$ of size $T_{i+1}[\mathcal S]$ which covers $\mathcal S$
and, together with $D_R$, dominates all of $G[\bigcup_{1\leq j \leq {i+1}} N[x_j]]$.
This shows that $|D_{i+1}| \leq T_{i+1}[\mathcal S]$ and we conclude that 
$|D_{i+1}| = T_{i+1}[\mathcal S]$, as claimed.

\end{document}